\documentclass[11pt]{article} 

\usepackage{a4wide}
\usepackage{algorithm}

\usepackage{algpseudocode}

\usepackage{fullpage}
\usepackage{amssymb} \usepackage{amsmath} \usepackage{bm}
\usepackage{times} \usepackage{fancybox} \usepackage{color}
\usepackage{latexsym}
\usepackage{eepic} \usepackage{epic} \usepackage{epsf}
\usepackage{graphics}  \usepackage{pstricks}
\psset{unit=.5pt} \newtheorem{theorem}{Theorem}[section]
\newtheorem{lemma}[theorem]{Lemma}

\newtheorem{fact}[theorem]{Fact}
\newtheorem{corollary}[theorem]{Corollary}
\newtheorem{claim}[theorem]{Claim}
\newtheorem{proposition}[theorem]{Proposition}

\newtheorem{definition}[theorem]{Definition}

\newtheorem{remark}[theorem]{Remark}  \newcommand{\bproof}{\noindent{\it Proof}}
\newcommand{\cproof}{\noindent{\it Proof of Claim}}
\newcommand{\eproof}{\hspace*{\fill}$\rule{2mm}{2mm}$~~~~~\bigskip}
\renewenvironment{proof}{\bproof. }{\eproof}

\newcommand{\PIT}{\mbox{\small\rm PIT}}

\newcommand{\Rad}{\mbox{\small\rm Rad}}
\newcommand{\Perm}{\ensuremath{Perm}}

\newcommand{\NEXP}{\mbox{\small\rm NEXP}}

  \newcommand{\NP}{\mbox{\rm NP}}

    \newcommand{\MA}{\mbox{\rm MA}}

  \renewcommand{\P}{\mbox{\rm P}}

\newcommand{\ppoly}{\mbox{\rm P/poly}}

 \newcommand{\poly}{\mbox{\rm
poly}}

\newcommand{\iso}{\ensuremath{\cong} }

 \newcommand{\Z}{\mathbb{Z}}
 
\newcommand{\F}{\ensuremath{\mathbb{F}}}

\renewcommand{\angle}[1]{\langle #1\rangle}

 \newcommand{\Prob}{\mbox{\rm Prob}}

 \newcommand{\R}{{R}}

\renewcommand{\mod}{\mbox{\textrm mod}}

\renewcommand{\angle}[1]{\langle #1\rangle}

\newcommand{\cA}{\mathcal{A}}

\title{New results on Noncommutative and Commutative Polynomial
  Identity Testing}

\author{V.~Arvind, Partha Mukhopadhyay, and Srikanth Srinivasan\\
Institute of Mathematical Sciences\\ C.I.T Campus,Chennai  600 113,
India\\
\tt{\{arvind,partham,srikanth\}@imsc.res.in}
}

\date{}

\begin{document}

\maketitle

\begin{abstract} 
  Using ideas from automata theory we design a new efficient
  (deterministic) identity test for the \emph{noncommutative}
  polynomial identity testing problem (first introduced and studied in
  \cite{RS05,BW05}). More precisely, given as input a noncommutative
  circuit $C(x_1,\cdots,x_n)$ computing a polynomial in
  $\F\{x_1,\cdots,x_n\}$ of degree $d$ with at most $t$ monomials,
  where the variables $x_i$ are noncommuting, we give a deterministic
  polynomial identity test that checks if $C\equiv 0$ and runs in time
  polynomial in $d, n, |C|$, and $t$.

  The same methods works in a black-box setting: Given a noncommuting
  black-box polynomial $f\in\F\{x_1,\cdots,x_n\}$ of degree $d$ with
  $t$ monomials we can, in fact, reconstruct the entire polynomial $f$
  in time polynomial in $n,d$ and $t$. Indeed, we apply this idea to
  the reconstruction of black-box noncommuting algebraic branching
  programs (the ABPs considered by Nisan in \cite{N91} and Raz-Shpilka
  in \cite{RS05}). Assuming that the black-box model allows us to
  query the ABP for the output at any given gate then we can
  reconstruct an (equivalent) ABP in deterministic polynomial time.

  Finally, we turn to commutative identity testing and explore the
  complexity of the problem when the coefficients of the input
  polynomial come from an arbitrary finite commutative ring with unity
  whose elements are uniformly encoded as strings and the ring
  operations are given by an oracle. We show that several algorithmic
  results for polynomial identity testing over fields also hold when
  the coefficients come from such finite rings.
\end{abstract}

\section{Introduction}\label{intro}

Polynomial identity testing (denoted $\PIT$) over fields is a well
studied algorithmic problem: given an arithmetic circuit $C$ computing
a polynomial in $\F[x_1, x_2, \cdots, x_n]$ over a field $\F$, the
problem is to determine whether the polynomial computed by $C$ is
identically zero. The problem is also studied when the input
polynomial $f$ is given only via black-box access. I.e.\ we can evaluate
it at any point in $\F^n$ or in $\F'^n$ for a field extension $\F'$ of
$\F$. When $f$ is given by a circuit the problem is in randomized
polynomial time. Even in the black-box setting, when $|\F|$ is
suitably larger than $\deg(f)$, the problem is in randomized
polynomial time. A major challenge it to obtain deterministic
polynomial time algorithms even for restricted versions of the
problem. The results of Impagliazzo and Kabanets \cite{KI03} show that
the problem is as hard as proving superpolynomial circuit lower
bounds. Indeed, the problem remains open even for depth-3 arithmetic
circuits with an unbounded $\Sigma$ gate as output \cite{DS05,KS07}.

As shown by Nisan \cite{N91} noncommutative algebraic computation is
somewhat easier to prove lower bounds. Using a rank argument Nisan has
shown exponential size lower bounds for noncommutative formulas
(and noncommutative algebraic branching programs) that compute the
noncommutative permanent or determinant polynomials in the ring
$\F\{x_1,\cdots,x_n\}$ where $x_i$ are noncommuting variables. Thus,
it seems plausible that identity testing in the noncommutative setting
ought to be easier too. Indeed, Raz and Shpilka in \cite{RS05} have
shown that that for noncommutative formulas (and algebraic branching
programs) there is a deterministic polynomial time algorithm for
polynomial identity testing. However, for noncommutative circuits the
situation is somewhat different.  Bogdanov and Wee in \cite{BW05} show
using Amitsur-Levitzki's theorem that identity testing for
\emph{polynomial degree} noncommutative circuits is in randomized
polynomial time.  Basically, the Amitsur-Levitzki theorem allows them
to randomly assign elements from a matrix algebra $M_k(\F)$ for the
noncommuting variables $x_i$, where $2k$ exceeds the degree of the
circuit.

The main contribution of this paper is the use of ideas from
\emph{automata theory} to design new efficient (deterministic)
polynomial identity tests for \emph{noncommutative} polynomials.  More
precisely, given a noncommutative circuit $C(x_1,\cdots,x_n)$
computing a polynomial of degree $d$ with $t$ monomials in
$\F\{x_1,\cdots,x_n\}$, where the variables $x_i$ are noncommuting, we
give a deterministic polynomial identity test that checks if $C\equiv
0$ and runs in time polynomial in $d, |C|, n$, and $t$. The main idea in
our algorithm is to think of the noncommuting monomials over the $x_i$
as words and to design finite automata that allow us to distinguish
between different words. Then, using the connection between automata,
monoids and matrix rings we are able to deterministically choose a
relatively small number of matrix assignments for the noncommuting
variables to decide if $C\equiv 0$. Thus, we are able to avoid using
the Amitsur-Levitzki theorem. Indeed, using our automata theory method
we can easily an alternative proof of (a weaker) version of
Amitsur-Levitzki which is good enough for algorithmic purposes as in
\cite{BW05} for example.

Our method actually works in a black-box setting. In fact, given a
noncommuting black-box polynomial $f\in \F\{x_1,\cdots,x_n\}$ of
degree $d$ with $t$ monomials, which we can evaluate by assigning
matrices to $x_i$, we can reconstruct the entire polynomial $f$ in
time polynomial in $n,d$ and $t$.

Furthermore, we also apply this idea to \emph{black-box} noncommuting
algebraic branching programs. We extend the result of Raz and Shpilka
\cite{RS05} by giving an efficient deterministic reconstruction
algorithm for black-box noncommuting algebraic branching programs
(wherein we are allowed to only query the ABP for input variables set
to matrices of polynomial dimension). Our black-box model assumes that
we can query for the output of \emph{any gate} of the ABP, not just the
output gate.

We now motivate and explain the other results in the paper.  Recently,
in \cite{AM07} we studied $\PIT$ (the usual commuting variables
setting) and its connection to the polynomial ideal membership
problem. Although ideal membership is EXPSPACE-complete, there is an
interesting similarity between the two problems which is the
motivation for the present paper. Suppose $I\subset
\F[x_1,\cdots,x_n]$ is an ideal generated by polynomials
$g_1,\cdots,g_r\in\F[x_1,\cdots,x_k]$ and $f\in \F[x_1,\cdots,x_n]$.
We observe that $f\in I$ if and only if $f$ is identically zero in the
ring $\F[x_1,\cdots,x_k]/I[x_{k+1},\cdots,x_n]$. Thus, ideal
membership is easily reducible to polynomial identity testing when the
coefficient ring is $\F[x_1,\cdots,x_k]/I$. Consequently, identity
testing for the coefficient ring $\F[x_1,\cdots,x_k]/I$ is
EXPSPACE-hard even when the polynomial $f$ is given explicitly as a
linear combination of monomials.

This raises the question about the complexity of $\PIT$ for a
polynomial ring $R[x_1,\cdots,x_n]$ where $R$ is a commutative ring
with unity. How does the complexity depend on the structure of the
ring $R$? We give a precise answer to this question in this paper.  We
show that the algebraic structure of $R$ is not important. It suffices
that the elements of $R$ have polynomial-size encoding, and w.r.t.
this encoding the ring operations can be efficiently performed. This
is in contrast to the ring $\F[x_1,\cdots,x_k]/I$: we have double
exponential number of elements of polynomial degree in
$\F[x_1,\cdots,x_k]$ and the ring operations in $\F[x_1,\cdots,x_k]/I$
are essentially ideal membership questions and hence computationally
hard.

More precisely, we study polynomial identity testing for \emph{finite}
commutative rings $R$, where we assume that the elements of $R$ are
uniformly encoded as strings in $\{0,1\}^m$ with two special strings
encoding $0$ and $1$, and the ring operations are carried out by
queries to the \emph{ring oracle}.

\section{Noncommutative Polynomial Identity Testing}\label{dit}

Recall that an \emph{arithmetic circuit} $C$ over a field 
$\F$ is defined as follows: $C$ takes as inputs, a set of
indeterminates (either commuting or noncommuting) and elements
from $\F$ as scalars. If $f,g$ are the inputs of an addition gate, then
the output will be $f+g$.  Similarly for a multiplication gate the
output will be $fg$. For noncommuting variables the circuit respect
the order of multiplication. An arithmetic circuit is a formula if the
fan-out of every gate is at most one.

Noncommutative identity testing was studied by Raz and Shpilka in
\cite{RS05} and Bogdanov and Wee in \cite{BW05}. In the Bogdanov-Wee
paper, they considered a polynomial $f$ of small degree over
$\F\{x_1,\cdots,x_n\}$, for a field $\F$, given by an arithmetic
circuit. They were able to give a randomized polynomial time
algorithm for the identity testing of $f$. The key feature of their
algorithm was a reduction from noncommutative identity testing to
commutative identity testing which is based on a classic theorem of
Amitsur and Levitzki~\cite{AL50} about minimal identities for
algebras.

Raz and Shpilka \cite{RS05} give a deterministic polynomial-time
algorithm for noncommutative formula identity testing by first
converting a homogeneous formula into a noncommutative algebraic
branching program (ABP), as done in \cite{N91}.

In this section we study the noncommutative polynomial identity
testing problem. Using simple ideas from automata theory, we design a
new deterministic identity test that runs in polynomial time if the
input circuit is sparse and of small degree. Our algorithm works with
only black-box access to the noncommuting polynomial, and we can even
efficiently reconstruct the polynomial.

We will first describe the algorithm to test if a sparse polynomial of
polynomial degree over noncommuting variables is identically zero.
Then we give an algorithm that reconstructs this sparse polynomial.
Though the latter result subsumes the former, for clarity of
exposition, we describe both. Furthermore, we note that we can assume
that the polynomial is given as an arithmetic circuit over a field
$\F$.

In the case of commuting variables, \cite{OT88} gives an
interpolation algorithm that computes the given sparse polynomial, and
thus can be used for identity testing. It is not clear how to
generalize this algorithm to the noncommutative setting. Our identity
testing algorithm evaluates the given polynomial at specific,
well-chosen points in a matrix algebra (of polynomial dimension over the
base field), such that any non-zero sparse polynomial is guaranteed to
evaluate to a non-zero matrix at one of these points. The
reconstruction algorithm uses the above identity testing algorithm as
a subroutine in a prefix-based search to find all the monomials and
their coefficients.

We now describe the identity testing algorithm informally.  Our idea
is to view each monomial as a short binary string. A sparse
polynomial, hence, is given by a polynomial number of such strings
(and the coefficients of the corresponding monomials). The algorithm
proceeds in two steps; in the first step, we construct a small set of
finite automata such that, given any small collection of short binary
strings, at least one automaton from the set accepts exactly one
string from this collection; in the second step, for each of the
automata constructed, we construct a tuple of points over a matrix
algebra over $\F$ such that the evaluation of any monomial at the
tuple `mimics' the run of the corresponding string on the automaton.
Now, given any non-zero polynomial $f$ of small degree with few terms,
we are guaranteed to have constructed an automaton $A$ `isolating' a
string from the collection of strings corresponding to monomials in
$f$. We then show that evaluating $f$ over the tuple corresponding to
$A$ gives us a non-zero output: hence, we can conclude $f$ is
non-zero. We now describe both algorithms formally.

\subsection{Preliminaries}\label{dit_prelim}

We first recall some standard automata theory notation (see, for
example, \cite{HU78}). Fix a finite automaton $A = (Q,\delta,q_0,q_f)$
which takes as input strings in $\{0,1\}^*$. $Q$ is the set of states
of $A$, $\delta:Q\times \{0,1\}\rightarrow Q$ is the transition
function, and $q_0$ and $q_f$ are the initial and final states
respectively (throughout, we only consider automata with unique
accepting states). For each letter $b\in\{0,1\}$, let
$\delta_b:Q\rightarrow Q$ be the function defined by: $\delta_b(q) =
\delta(q,b)$. These functions generate a submonoid of the monoid of
all functions from $Q$ to $Q$. This is the transition monoid of the
automaton $A$ and is well-studied in automata theory: for example, see
\cite[page 55]{Str94}. We now define the $0$-$1$ matrix $M_b \in
\F^{|Q|\times |Q|}$ as follows:
\[
M_b(q,q') = \left\{\begin{array}{cc}
1 & \textrm{if $\delta_b(q) = q'$,}\\
0 & \textrm{otherwise.} \end{array}\right.
\]

The matrix $M_b$ is simply the adjacency matrix of the graph of the
function $\delta_b$. As the entries of $M_b$ are only zeros and ones,
we can consider $M_b$ to be a matrix over any field $\F$.

Furthermore, for any $w=w_1w_2\cdots w_k\in\{0,1\}^*$ we define the
matrix $M_w$ to be the matrix product $M_{w_1}M_{w_2}\cdots M_{w_k}$.
If $w$ is the empty string, define $M_w$ to be the identity matrix of
dimension $|Q|\times |Q|$. For a string $w$, let $\delta_w$ denote the
natural extension of the transition function to $w$; if $w$ is the
empty string, $\delta_w$ is simply the identity function. It is easy
to check that:
\begin{equation}
\label{dit_prelim_eqn}
M_w(q,q') = \left\{\begin{array}{cc}
1 & \textrm{if $\delta_w(q) = q'$,}\\
0 & \textrm{otherwise.} \end{array}\right.
\end{equation}
Thus, $M_w$ is also a matrix of zeros and ones for any string $w$.
Also, $M_w(q_0,q_f) = 1$ if and only if $w$ is accepted by the
automaton $A$.

\subsection{The output of a circuit on an automaton}\label{dit_output}

Now, we consider the ring $\F\{x_1,\cdots,x_n\}$ of polynomials with
noncommuting variables $x_1,\cdots,x_n$ over a field $\F$. Let $C$ be
a noncommutative arithmetic circuit computing a polynomial $f\in
\F\{x_1,\cdots,x_n\}$. Let $d$ be an upper bound on the degree of
$f$. We can consider monomials over the noncommuting variables
$x_1,\cdots,x_n$ as strings over an alphabet of size $n$. For our
construction in Section~\ref{dit_auto_crt}, it is convenient to encode
the variables $x_i$ in the alphabet $\{0,1\}$. We do this by encoding
the variable $x_i$ by the string $v_i=01^i0$, which is basically a
unary encoding with delimiters. Clearly, each monomial over the
$x_i$'s of degree at most $d$ maps uniquely to a binary string of
length at most $d(n+2)$.

Let $A = (Q,\delta,q_0,q_f)$ be a finite automaton over the alphabet
$\{0,1\}$. With respect to automaton $A$ we have matrices $M_{v_i}\in
\F^{|Q|\times |Q|}$ as defined in Section \ref{dit_prelim}, where each
$v_i$ is the binary string that encodes $x_i$. We are interested in
the output matrix obtained when the inputs $x_i$ to the circuit $C$
are replaced by the matrices $M_{v_i}$. This output matrix is defined
in the obvious way: the inputs are $|Q|\times |Q|$ matrices and we do
matrix addition and matrix multiplication at each addition (resp.
multiplication) of the circuit $C$. We define the \emph{output of $C$
  on the automaton $A$} to be this output matrix $M_{out}$. Clearly,
given circuit $C$ and automaton $A$, the matrix $M_{out}$ can be
computed in time $\poly(|C|,|A|,n)$.

We observe the following property: the matrix output $M_{out}$ of $C$
on $A$ is determined completely by the polynomial $f$ computed by $C$;
the structure of the circuit $C$ is otherwise irrelevant. This is
important for us, since we are only interested in $f$. In particular,
the output is always $0$ when $f\equiv 0$.

More specifically, consider what happens when $C$ computes a
polynomial with a single term, say $f(x_1,\cdots,x_n) = cx_{j_1}\cdots
x_{j_k}$, with a non-zero coefficient $c\in\F$. In this case, the
output matrix $M_{out}$ is clearly the matrix $cM_{v_{j_1}}\cdots
M_{v_{j_k}} = cM_w$, where $w=v_{j_1}\cdots v_{j_k}$ is the binary
string representing the monomial $x_{j_1}\cdots x_{j_k}$. Thus, by
Equation \ref{dit_prelim_eqn} above, we see that the entry
$M_{out}(q_0,q_f)$ is $0$ when $A$ rejects $w$, and $c$ when $A$
accepts $w$. In general, suppose $C$ computes a polynomial $f =
\sum_{i=1}^{t}c_i m_i$ with $t$ nonzero terms, where $c_i\in
\F\setminus\{0\}$ and $m_i = \prod_{j=1}^{d_i}x_{i_j}$, where $d_i\leq
d$. Let $w_i = v_{i_1}\cdots v_{i_{d_i}}$ denote the binary string
representing monomial $m_i$.  Finally, let $S_A^f =
\{i\in\{1,\cdots,t\}~|~A \textrm{ accepts } w_i\}$.

\begin{theorem}\label{autothm}
Given any arithmetic circuit $C$ computing polynomial $f\in
\F\{x_1,\cdots,x_n\}$ and any finite automaton $A =
(Q,\delta,q_0,q_f)$, then the output $M_{out}$ of $C$ on $A$ is such
that $M_{out}(q_0,q_f) = \sum_{i\in S_A^f}c_i$.
\end{theorem}

\begin{proof}
  The proof is an easy consequence of the definitions and the
  properties of the matrices $M_w$ stated in Section \ref{dit_prelim}.
  Note that $M_{out} = f(M_{v_1},\cdots,M_{v_n})$. But 
  $f(M_{v_1},\cdots,M_{v_n}) = \sum_{i=1}^{s}c_i M_{w_i}$, where $w_i
  = v_{i_1}\cdots v_{i_{d_i}}$ is the binary string representing
  monomial $m_i$. By Equation \ref{dit_prelim_eqn}, we know that
  $M_{w_i}(q_0,q_f)$ is $1$ if $w_i$ is accepted by $A$, and $0$
  otherwise. Adding up, we obtain the result.
\end{proof}

We now explain the role of the automaton $A$ in testing if the
polynomial $f$ computed by $C$ is identically zero or not. Our basic
idea is to try and design an automaton $A$ that accepts exactly one
word from among all the words that correspond to the non-zero terms in
$f$. This would ensure that $M_{out}(q_0,q_f)$ is the non-zero
coefficient of the monomial filtered out. More precisely, we will use
the above theorem primarily in the following form, which we state as a
corollary.

\begin{corollary}
Given any arithmetic circuit $C$ computing polynomial $f\in
\F\{x_1,\cdots,x_n\}$ and any finite automaton $A =
(Q,\delta,q_0,q_f)$, then the output $M_{out}$ of $C$ on $A$ satisfies:
\label{autocor_use}
\begin{itemize}
\item[(1)] If $A$ rejects every string corresponding to a monomial in
$f$, then $M_{out}(q_0,q_f) = 0$.  
\item[(2)] If $A$ accepts exactly one string corresponding to a
monomial in $f$, then $M_{out}(q_0,q_f)$ is the nonzero coefficient of
that monomial in $f$.
\end{itemize}
Moreover, $M_{out}$ can be computed in time $\poly(|C|,|A|,n)$.
\end{corollary}

\begin{proof}
Both points ($1$) and ($2$) are immediate consequences of the above
theorem. The complexity of computing $M_{out}$ easily follows from its
definition.
\end{proof}\\
Another interesting corollary to the above theorem is the following.

\begin{corollary}\label{autocor_pick}
  Given any arithmetic circuit $C$ over $\F\{x_1,\cdots,x_n\}$, and
  any monomial $m$ of degree $d_m$, we can compute the coefficient of
  $m$ in $C$ in time $\poly(|C|,d_m,n)$.
\end{corollary}

\begin{proof}
  Apply Corollary \ref{autocor_use} with $A$ being any standard
  automaton that accepts the string corresponding to monomial $m$ and
  rejects every other string. Clearly, $A$ can be chosen so that $A$
  has a unique accepting state and $|A| = O(nd_m)$.
\end{proof}

\begin{remark}
  Observe that Corollary~\ref{autocor_pick} is highly unlikely to hold
  in the commutative setting $\F[x_1,\cdots,x_n]$. For, in the
  commutative case, computing the coefficient of the monomial
  $x_1\cdots x_n$ in even an arbitrary product of linear forms
  $\Pi_{i}\ell_i$ is at least as hard as the permanent problem over
  $\F$, which is \mbox{$\#$\rm P}-complete when $\F=\mathbb{Q}$.
\end{remark}

\begin{remark}
  Corollary \ref{autocor_use} can also be used to give an independent
  proof of a weaker form of the result of Amitsur and Levitzki that is
  stated in Lemma \ref{AL-1}. In particular, it is easy to see that
  the algebra $M_d(\F)$ of $d\times d$ matrices over the field $\F$
  does not satisfy any nontrivial identity of degree $<d$. To prove
  this, we will consider noncommuting monomials as strings directly
  over the $n$ letter alphabet $\{x_1,\cdots,x_n\}$. Suppose
  $f=\sum_{i=1}^t c_im_i \in\F\{x_1,\cdots,x_n\}$ is a nonzero
  polynomial of degree $<d$.  Clearly, we can construct an automaton
  $B$ over the alphabet $\{x_1,\cdots,x_n\}$ that accepts exactly one
  string, namely one nonzero monomial, say $m_{i_0}$, of $f$ and
  rejects all the other strings over $\{x_1,\cdots,x_n\}$.  Also,
  $B$ can be constructed with at most $d$ states. Now, consider the
  output $M_{out}$ of any circuit computing $f$ on $B$. By Corollary
  \ref{autocor_use} the output matrix is non-zero, and this proves the
  result.
\end{remark}

\subsection{Construction of finite automata}\label{dit_auto_crt}

We begin with a useful definition.

\begin{definition}\label{isolate}{\hfill{~}}
  Let $W$ be a finite set of binary strings and $\mathcal{A}$ be a
  finite family of finite automata over the binary alphabet $\{0,1\}$.
\begin{itemize}
\item We say that $\mathcal{A}$ is \emph{isolating} for $W$ if there
  exists a string $w\in W$ and an automaton $A\in\mathcal{A}$ such
  that $A$ accepts $w$ and rejects all $w'\in W\setminus\{w\}$.
\item We say that $\mathcal{A}$ is an \emph{$(m,s)$-isolating family}
  if for every subset $W=\{w_1,\cdots,w_s\}$ of $s$ many binary
  strings, each of length at most $m$, there is a $A\in\mathcal{A}$ 
  such that $A$ is isolating for $W$. 
\end{itemize}
\end{definition}

Fix parameters $m,s\in\mathbb{N}$. Our first aim is to construct an
$(m,s)$ isolating family of automata $\mathcal{A}$, where both
$|\mathcal{A}|$ and the size of each automaton in $\mathcal{A}$ is
polynomially bounded in size. Then, combined with Corollary
\ref{autocor_use} we will be able to obtain deterministic identity
testing and interpolation algorithms in the sequel. 

Recall that we only deal with finite automata that have unique
accepting states. In what follows, for a string $w\in\{0,1\}^*$, we
denote by $n_w$ the positive integer represented by the binary numeral
$1w$. For each prime $p$ and each integer $i\in\{0,\cdots,p-1\}$, we
can easily construct an automaton $A_{p,i}$ that accepts exactly those
$w$ such that $n_w\equiv i~(\mod~p)$. Moreover, $A_{p,i}$ can be
constructed so as to have $p$ states and exactly one final state.

Our collection of automata $\mathcal{A}$ is just the set of $A_{p,i}$
where $p$ runs over the first few polynomially many primes, and
$i\in\{0,\cdots,p-1\}$. Formally, let $N$ denote
$(m+2)\binom{s}{2}+1$; $\mathcal{A}$ is the collection of $A_{p,i}$,
where $p$ runs over the first $N$ primes and $i\in\{0,\cdots,p-1\}$.
Notice that, by the prime number theorem, all the primes chosen above
are bounded in value by $N^2$, which is clearly polynomial in $m$ and
$s$. Hence, $|\mathcal{A}| = \poly(m,s)$, and each $A\in\mathcal{A}$
is bounded in size by $\poly(m,s)$. In the following lemma we
show that $\mathcal{A}$ is an $(m,s)$-isolating automata family.

\begin{lemma}\label{automata_crt}
  The family of finite automata $\mathcal{A}$ defined as above
is an $(m,s)$-isolating automata family.
\end{lemma}

\begin{proof}
  Consider any set of $s$ binary strings $W$ of length at most $m$
  each. By the construction of $\mathcal{A}$, $A_{p,i}\in\mathcal{A}$
  isolates $W$ if and only if $p$ does not divide
  $n_{w_j}-n_{w_k}$ for some $j$ and all $k\neq j$, and $n_{w_j}\equiv
  i~(\mod~p)$. Clearly, if $p$ satisfies the first of these
  conditions, $i$ can easily be chosen so that the second condition is
  satisfied. We will show that there is some prime among the first $N$
  primes that does not divide $P = \prod_{j\neq k}(n_{w_j}-n_{w_k})$.
  This easily follows from the fact that the number of distinct prime
  divisors of $P$ is at most $\log{|P|}$, which is clearly bounded by
  $(m+2)\binom{s}{2} = N-1$.  This concludes the proof.
\end{proof}

We note that the above $(m,s)$-isolating family $\mathcal{A}$ can 
clearly be constructed in time $\poly(m,s)$.

\subsection{The identity testing algorithm}
\label{dit_algo_it}
We now describe the identity testing algorithm. Let $C$ be the input
circuit computing a polynomial $f$ over $\F\{x_1,\cdots,x_n\}$.
Let $t$ be an upper bound on the number of monomials in $f$, and $d$
be an upper bound on the degree of $f$. As in Section \ref{dit_output},
we represent monomials over $x_1,\cdots,x_n$ as binary strings. Every
monomial in $f$ is represented by a string of length at most $d(n+2)$.

Our algorithm proceeds as follows: Using the construction of Section
\ref{dit_auto_crt}, we compute a family $\mathcal{A}$ of automata such that
$\mathcal{A}$ is isolating for any set $W$ with at most $t$ strings 
of length at most
$d(n+2)$ each. For each $A\in\mathcal{A}$, the algorithm computes the
output $M_{out}$ of $C$ on $A$. If $M_{out} \neq 0$ for any $A$, then
the algorithm concludes that the polynomial computed by the input
circuit is not identically zero; otherwise, the algorithm declares
that the polynomial is identically zero.

The correctness of the above algorithm is almost immediate from
Corollary \ref{autocor_use}. If the polynomial is identically zero, it is
easy to see that the algorithm outputs the correct answer. If the
polynomial is nonzero, then by the construction of $\mathcal{A}$, we
know that there exists $A\in\mathcal{A}$ such that $A$ accepts
precisely one of the strings corresponding to the monomials in $f$.
Then, by Corollary \ref{autocor_use}, the output of $C$ on $A$ is nonzero.
Hence, the algorithm correctly deduces that the polynomial computed is
not identically zero.

As for the running time of the algorithm, it is easy to see that the
family of automata $\mathcal{A}$ can be constructed in time
$\poly(d,n,t)$. Also, the matrices $M_{v_i}$ for each $A$ (all of
which are of size $\poly(d,n,t)$) can be constructed in polynomial
time. Hence, the entire algorithm runs in time $\poly(|C|,d,n,t)$. We
have proved the following theorem:

\begin{theorem}\label{sparsetest}
Given any arithmetic circuit $C$ with the promise that $C$ computes a
polynomial $f\in \F\{x_1,\cdots,x_n\}$ of degree $d$ with at most $t$
monomials, we can check, in time $\poly(|C|,d,n,t)$, if $f$ is
identically zero. In particular, if $f$ is sparse and of polynomial
degree, then we have a deterministic polynomial time algorithm.
\end{theorem}

In the case of arbitrary noncommutative arithmetic circuits,
\cite{BW05} gives a randomized exponential time algorithm for the
identity testing problem. Their algorithm is based on the
Amitsur-Levitzki theorem, which forces the identity test to randomly
assign exponential size matrices for the noncommuting variables since
the circuit could compute an exponential degree polynomial.  However,
notice that Theorem~\ref{sparsetest} gives a deterministic
exponential-time algorithm under the additional restriction that the
input circuit computes a polynomial with at most \emph{exponentially}
many monomials. In general, a polynomial of exponential degree can
have a double exponential number of terms.

\subsection{Interpolation of noncommutative polynomials}\label{dit_algo_inter}

We now describe an algorithm that efficiently computes the
noncommutative polynomial given by the input circuit. Let $C,f,t$ and
$d$ be as in Section \ref{dit_algo_it}. Let $W$ denote the set of all
strings corresponding to monomials with non-zero coefficients in $f$.
For all binary strings $w$, let $A_w$ denote any standard automaton
that accepts $w$ and rejects all other strings. For any automaton $A$
and string $w$, we let $[A]_w$ denote the automaton that accepts those
strings that are accepted by $A$ and in addition, contain $w$ as a
prefix. For a set of finite automata $\mathcal{A}$, let
$[\mathcal{A}]_w$ denote the set $\{[A]_w~|~A\in\mathcal{A}\}$. 

We now describe a subroutine {\tt Test} that takes as input an
arithmetic circuit $C$ and a set of finite automata $\mathcal{A}$ and
returns a field element $\alpha\in \F$. The subroutine {\tt Test} will 
have the following properties: 
\begin{itemize}
\item[(P1)]If $\cA$ is isolating for $W$, the set of strings corresponding to
monomials in $f$, then $\alpha\neq0$.  
\item[(P2)] In the special case when $|\cA|=1$, and the above holds,
then $\alpha$ is in fact the coefficient of the isolated monomial.  
\item[(P3)] If no $A\in\cA$ accepts any string in $W$, then
$\alpha=0$.  
\end{itemize}

We now give the easy description of {\tt Test($C$,$\cA$)}: 

For each $A\in\cA$, the subroutine ${\tt Test}$ computes the output
matrix $M_{out}^A$ of $C$ on $A$. If there is an $A\in\cA$ such that
$M_{out}^A(q_0^A,q_f^A)\neq 0$, then for the first such automaton
$A\in\cA$, ${\tt Test}$ returns the scalar
$\alpha=M_{out}^A(q_0^A,q_f^A)$. Here, notice that $q_0^A$, $q_f^A$
denote the initial and final states of the automaton $A$. If there is
no such automaton $A\in\cA$ is found, then the subroutine returns the
scalar $0$.

It follows directly from Corollary \ref{autocor_use} that {\tt Test}
has Properties (P1)-(P3). Furthermore, clearly {\tt Test} runs in time
$\poly(|C|,||\cA||)$, where $||\cA||$ denotes the sum of the sizes of
the automata in $\cA$.

Let $f\in\F\{x_1,\cdots,x_n\}$ denote the noncommuting polynomial
computed by the input circuit $C$. We now describe a recursive
prefix-search based algorithm ${\tt Interpolate}$ that takes as input
the circuit $C$ and a binary string $u$, and computes all those
monomials of $f$ (along with their coefficients) which contain $u$ as
a prefix when encoded as strings using our encoding $x_i\mapsto
v_i=01^i0$. Clearly, in order to obtain all monomials of $f$ with
their coefficients, it suffices to run this algorithm with
$u=\epsilon$, the empty string.

In what follows, let $\cA_0$ denote the $(m,s)$-isolating automata
family $\{A_{p,i}\}$ as constructed in Section \ref{dit_auto_crt} with
parameters $m=d(n+2)$ and $s=t$. As explained in Section
\ref{dit_auto_crt}, we can compute $\cA_0$ in time $\poly(d,n,t)$.

Suppose $f$ is the polynomial computed by the circuit $C$. We 
now describe the algorithm {\tt Interpolate($C$,$u$)} formally (Algorithm 1).  
\begin{algorithm}
\caption{The Interpolation algorithm}
\begin{algorithmic}[1]

\Procedure {\tt Interpolate}{$C$,$u$}
 \State $\alpha,\alpha',\alpha''\leftarrow 0$.
  \State $\alpha\leftarrow $ {\tt Test$(C,\{A_u\})$}\Comment{$A_u$ is 
  the standard automaton that accepts only $u$} 
   \If {$\alpha=0$}
     \State \textbf{Break}. \Comment{$u$ can not corresponds 
   to a monomial of $f$}
    \Else \State \textbf{Output} $(u,\alpha)$. \Comment{$u$ is the binary 
encoding of a monomial of $f$ with coefficient $\alpha$} 
    \EndIf

 Now the algorithm find all monomials (along with their
 coefficient) 
 
 containing $u0$ or $u1$ as prefix in the binary encoding.

    \If {$|u|=d(n+2)$}
      \State \textbf{Stop}. 
    \Else 
      \State $\alpha'\leftarrow${\tt Test$(C,[\cA_0]_{u0})$}, 
             $\alpha''\leftarrow${\tt Test$(C,[\cA_0]_{u1})$}. 

    \EndIf 
    
    \If {$\alpha'\neq 0$}
        \State {\tt Interpolate$(C,u0)$}. \Comment{There is some monomial in 
    $C$ extending $u0$}
    \EndIf    
   
    \If {$\alpha''\neq 0$}
        \State {\tt Interpolate$(C,u1)$}. \Comment{There is some monomial in 
    $C$ extending $u1$}
    \EndIf
\EndProcedure
\end{algorithmic}
\end{algorithm}

The correctness of this algorithm is clear from the correctness
of the {\tt Test} subroutine and Lemma \ref{automata_crt}.  To bound
the running time, note that the algorithm never calls ${\tt
Interpolate}$ on a string $u$ unless $u$ is the prefix of some string
corresponding to a monomial. Hence, the algorithm invokes ${\tt
  Interpolate}$ for at most $O(td(n+2))$ many prefixes $u$. Since
$||[\cA_0]_{u0}||$ and $|A_u|$ are both bounded by $\poly(d,n,t)$ for
all prefixes $u$, it follows that the running time of the algorithm is
$\poly(|C|,d,n,t)$. We summarize this discussion in the following
theorem.

\begin{theorem}\label{dit_main}
  Given any arithmetic circuit $C$ computing a polynomial $f\in
  \F\{x_1,\cdots,x_n\}$ of degree at most $d$ and with at most $t$
  monomials, we can compute all the monomials of $f$, and their
  coefficients, in time $\poly(|C|,d,n,t)$. In particular, if $C$
  computes a sparse polynomial $f$ of polynomial degree, then $f$ can
  be reconstructed in polynomial time.
\end{theorem}

\section{Interpolation of Algebraic Branching Programs over
  noncommuting variables}

In this section, we study the interpolation problem for black-box
Algebraic Branching Programs (ABP) computing a polynomial in the
noncommutative ring $\mathbb{F}\{x_1,\cdots,x_n\}$. We are given as
input an ABP (defined below) $P$ in the black-box setting, and our
task is to output an ABP $P'$ that computes the same polynomial as
$P$. To make the task feasible in the black-box setting, we assume
that we are allowed to evaluate $P$ at any of its intermediate gates.

We first observe that all the results in Section \ref{dit} hold under
the assumption that the input polynomial $f$ is allowed only
\emph{black-box access}. In the noncommutative setting, we shall
assume that the black-box access allows the polynomial to be evaluated
for input values from an arbitrary matrix algebra over the base field
$\F$. It is implicit here that the cost of evaluation is polynomial in
the dimension of the matrices. Note that this is a reasonable
noncommutative black-box model, because if we can evaluate $f$ only
over $\F$ or any commutative extension of $\F$, then we cannot
distinguish the non-commutative polynomial represented by $f$ from the
corresponding commutative polynomial. We state the black-box version
of our results below.

\begin{theorem}[Similar to Theorem \ref{autothm}]
\label{autothm_bb}
Given black-box access to any polynomial $f = \sum_{i=1}^{t}c_im_i\in
\F\{x_1,\cdots,x_n\}$ and any finite automaton $A =
(Q,\delta,q_0,q_f)$, then the output $M_{out}$ of $f$ on $A$ is such
that $M_{out}(q_0,q_f) = \sum_{i\in S_A^f}c_i$, where $S_A^f =
\{i~|~1\leq i\leq t \textrm{ and $A$ accepts the string $w_i$
corresponding to $m_i$}\}$
\end{theorem}
Here the output of polynomial $f$ on $A$ is defined analogously to the
output of a circuit on $A$ in Section \ref{dit_output}.

\begin{corollary}[Similar to Corollary \ref{autocor_pick}]
\label{autocor_pick_bb}
Given black-box access to a polynomial $f$ in $\F\{x_1,\cdots,x_n\}$,
and any monomial $m$ of degree $d_m$, we can compute the coefficient
of $m$ in $f$ in time $\poly(d_m,n)$.
\end{corollary}
Finally we have,

\begin{theorem}[Similar to Theorem \ref{dit_main}]
  Given black-box access to a polynomial $f$ in $\F\{x_1,\cdots,x_n\}$
  of degree at most $d$ and with at most $t$ monomials, we can compute
  all the monomials of $f$, and their coefficients, in time
  $\poly(d,n,t)$. In particular, if $f$ is a sparse polynomial of
  polynomial degree, then it can be reconstructed in polynomial time.
\end{theorem}

Our interpolation algorithm for noncommutative ABPs is motivated by
Raz and Shpilka's \cite{RS05} algorithm for identity testing of ABPs
over noncommuting variables. Our algorithm interpolates the given ABP
layer by layer using ideas developed in Section \ref{dit} (principally
Corollary \ref{autocor_pick_bb}).

\begin{definition}\cite{N91, RS05}
An Algebraic Branching Program (ABP) is a directed acyclic graph with
one vertex of in-degree zero, called the source, and a vertex of
out-degree zero, called the sink. The vertices of the graph are
partitioned into levels numbered $0,1,\cdots,d$. Edges may only go
from level $i$ to level $i+1$ for $i\in\{0,\cdots,d-1\}$. The source
is the only vertex at level $0$ and the sink is the only vertex at
level $d$. Each edge is labeled with a homogeneous linear form in the
input variables. The size of the ABP is the number of vertices.
\end{definition}

Notice that an ABP with no edge between two vertices $u$ and $v$ on
levels $i$ and $i+1$ is equivalent to an ABP with an edge from $u$ to
$v$ labeled with the zero linear form. Thus, without loss of
generality, we assume that in the given ABP there is an edge between
every pair of vertices on adjacent levels.

As mentioned before, we will assume black-box access to the input ABP
$P$ where we can evaluate the polynomial computed by $P$ at any of its
gates over arbitrary matrix rings over $\mathbb{F}$. In order to
specify the gate at which we want the output, we index the gates of
$P$ with a layer number and a gate number (in the layer).

Based on \cite{RS05}, we now define a \emph{Raz-Shpilka basis} for
the level $i$ of the ABP. Let the number of nodes at the $i$-th level
be $G_i$ and let $\{p_1,p_2,\cdots,p_{G_i}\}$ be the polynomials
computed at the nodes. We will identify this set of polynomials with
the $G_i\times n^i$ matrix $M_i$ where the columns of $M_i$ are
indexed by $n^i$ different monomials of degree $i$, and the rows are
indexed by the polynomials $p_j$. The entries of the matrix $M_i$ are
the corresponding polynomial coefficients. A Raz Shpilka basis is a
set of at most $G_i$ linearly independent column vectors of $M_i$ that
generates the entire column space. Notice that every vector in the
basis is identified by a monomial.

In the algorithm we need to compute a Raz-Shpilka basis at every
level of the ABP. Notice that at the level $0$ it is trivial to
compute such a basis. Inductively assume we can compute such a basis
at the level $i$. Denote the basis by $B_i=\{v_1,v_2,\cdots,v_{k_i}\}$
where $v_j\in\F^{G_i}$, and $k_i\leq G_i$. Assume that the elements of
this basis corresponds to the monomials $\{m_1,m_2,\cdots,m_{k_i}\}$.
We compute a Raz Shpilka basis at the level $i+1$ by computing the
column vectors corresponding to the set of monomials
$\{m_jx_s\}_{j\in[k_i],s\in[n]}$ in $M_{i+1}$ and then extracting the
linear independent vectors out of them.  Computing these column
vectors requires the computation of the coefficients of these
monomials, which can be done in polynomial time using the
Corollary~\ref{autocor_pick_bb}. Notice that we also know the
monomials that the elements of this basis correspond to.

We now describe the interpolation algorithm formally. As mentioned
before, we will construct the output ABP $P'$ layer by layer such that
every gate of $P'$ computes the same polynomial as the corresponding
gate in $P$. Clearly, this task is trivial at level $0$. 

Assume that we have completed the construction up to level $i<d$. We
now construct level $i+1$. This only involves computation of the linear
forms between level $i$ and level $i+1$. Hence, there are $k_i\leq
G_i$ vectors in the Raz-Shpilka basis at the $i$th level. Let the
monomials corresponding to these vectors be $B =
\{m_1,\cdots,m_{k_i}\}$. Fix any gate $u$ at level $i+1$ in $P$, and
let $p_u$ be the polynomial compute at this gate in
$P$. Clearly,
\[p_u = \sum_{j=1}^{G_i}p_j \ell_j\]
where $p_j$ is the polynomial computed at the $j$th gate at level $i$,
and $\ell_j$ is the linear form labeling the edge between the $j$th
gate at level $i$ and $u$.

We have,
\begin{align*}
p_u &= \sum_{j=1}^{G_i}p_j \ell_j\\
    &=
    \sum_{j=1}^{G_i}\left(\sum_{m:|m|=i}c_m^{(j)}m\right)\left(\sum_{s=1}^{n}a_s^{(j)}x_s\right)\\
    &=
    \sum_{m:|m|=i,s}mx_s\left(\sum_{j=1}^{G_i}c_m^{(j)}a_s^{(j)}\right)\\
    &= \sum_{m:|m|=i,s}mx_s\angle{c_m,a_s}
\end{align*}
where $c_m$ and $a_s$ denote the vectors of field elements
$(c_m^{(j)})_j$ and $(a_s^{(j)})_j$ respectively. Note that $a_s$
denotes a vector of unknowns that we need to compute. Each monomial
$mx_s$ in the above equation gives us a linear constraint on $a_s$.
However, this system of constraints is exponential in size. To obtain
a feasible solution for $\{a_s\}_{s\in [n]}$, we observe that it is
sufficient to satisfy the constraints corresponding only to monomials
$mx_s$ where $m\in B$. All other constraints are simply linear
combinations of these and are thus automatically satisfied by any
solution to these.

Now, for $m\in B$ and $s\in\{1,\cdots,n\}$, we compute the
coefficients of $mx_s$ in $p_u$ and those of $m$ in each of the $p_i$'s
using the algorithm of Corollary \ref{autocor_pick_bb}. Hence, we have
all the linear constraints we need to solve for $\{a_s\}_{s\in [n]}$.
Firstly, note that such a solution exists, since the linear forms in
the black box ABP $P$ give us such a solution. Moreover, any solution
to this system of linear equations generates the same polynomial $p_u$
at gate $u$.  Hence, we can use any solution to this system of linear
equations as our linear forms. We perform this computation for all
gates $u$ at the $i+1$st level. The final step in the iteration is to
compute the Raz-Shpilka basis for the level $i+1$. 

We can use induction on the level numbers to argue correctness of the
algorithm. From the input black-box ABP $P$, for each level $k$, let
$P_{jk}, 1\leq j\leq G_k$ denote the algebraic branching programs
computed by $P$ with output gate as gate $j$ in level $k$.  Assume, as
induction hypothesis, that the algorithm has computed linear forms for
all levels upto level $i$ and, furthermore, that the algorithm has a
correct Raz-Shpilka basis for all levels upto level $i$. This gives us
a reconstructed ABP $P'$ upto level $i$ with the property, for $1\leq
k\leq i$, each ABP $P'_{jk}, 1\leq j\leq G_k$ computes the same
polynomials as the corresponding $P_{jk}, 1\leq j\leq G_k$, where
$P'_{jk}$ is obtained from $P'$ by designating gate $j$ at level $k$
as output gate. Under this induction hypothesis, it is clear that our
interpolation algorithm will compute a correct set of linear forms
between levels $i$ and $i+1$. Consequently, the algorithm will
correctly reconstruct an ABP $P'$ upto level $i+1$ along with a
corresponding Raz-Shpilka basis for that level.

We can now summarize the result in the following theorem.

\begin{theorem}
  Let $P$ be an ABP of size $s$ and depth $d$ over
  $\F\{x_1,x_2,\cdots,x_n\}$ given by black-box access that allows
  evaluation of any gate of $P$ for inputs $x_i$ chosen from a matrix
  algebra $M_k(\F)$ for a polynomially bounded value of $k$. Then in
  deterministic time $\poly(d,s,n)$, we can compute an ABP $P'$ such
  that $P'$ evaluates to the same polynomial as $P$.
\end{theorem}  

 \section{Noncommutative identity testing and circuit lower bounds}

 In Section~\ref{dit} we gave a new deterministic identity test for
 noncommuting polynomials which runs in polynomial time for sparse
 polynomials of polynomially bounded degree.

 However, the real problem of interest is identity testing for
 polynomials given by small degree noncommutative circuits for which
 Bogdanov and Wee \cite{BW05} give an efficient randomized test. When
 the noncommutative circuit is a formula, Raz and Shpilka \cite{RS05}
 have shown that the problem is in deterministic polynomial
 time. Their method uses ideas from Nisan's lower bound technique for
 noncommutative formulae \cite{N91}.

 How hard would it be to show that noncommutative PIT is in
 deterministic polynomial time for \emph{circuits} of polynomial
 degree? In the commutative case, Impagliazzo and Kabanets \cite{KI03}
 have shown that derandomizing PIT implies circuit lower bounds. It
 implies that either $\NEXP\not\subseteq \ppoly$ or the integer
 Permanent does not have polynomial-size arithmetic circuits.

 We observe that this result also holds in the noncommutative setting.
 I.e., if noncommutative PIT has a deterministic polynomial-time
 algorithm then either $\NEXP\not\subseteq \ppoly$ or the
 \emph{noncommutative} Permanent function does not have
 polynomial-size noncommutative circuits.

 As noted, in some cases noncommutative circuit lower bounds are
 easier to prove than for commutative circuits. Nisan \cite{N91} has
 shown exponential-size lower bounds for noncommutative formula size
 and further results are known for pure noncommutative circuits
 \cite{N91,RS05}. However, proving superpolynomial size lower bounds for
 general noncommutative circuits computing the Permanent has remained
 an open problem.

The noncommutative Permanent function $\Perm(x_1,\cdots,x_n)\in
R\{x_1,\cdots,x_n\}$ is defined as
\[
\Perm(x_1,\cdots,x_n)=\sum_{\sigma\in S_n}\prod_{i=1}^n x_{i,\sigma(i)},
\]
where the coefficient ring $R$ is any commutative ring with unity.
Specifically, for the next theorem we choose $R=\mathbb{Q}$.

\begin{theorem}
  If $\PIT$ for noncommutative circuits of polynomial degree
  $C(x_1,\cdots,x_n)\in\mathbb{Q}\{x_1,\cdots,x_n\}$ is in
  deterministic polynomial-time then either
  $\NEXP\not\subseteq\ppoly$ or the \emph{noncommutative} Permanent
  function does not have polynomial-size noncommutative circuits.
\end{theorem}

\begin{proof}
  Suppose $\NEXP\subseteq\ppoly$.  Then, by the main result of
  \cite{IKW02} we have $\NEXP=\MA$.  Furthermore, by Toda's theorem
  $\MA\subseteq\P^{\Perm_{\mathbb{Z}}}$, where the oracle computes the
  integer permanent. Now, assuming $\PIT$ for noncommutative circuits
  of polynomial degree is in deterministic polynomial-time we will
  show that the (noncommutative) Permanent function does not have
  polynomial-size noncommutative circuits.  Suppose to the contrary
  that it does have polynomial-size noncommutative circuits. Clearly,
  we can use it to compute the integer permanent as well. Furthermore,
  as in \cite{KI03} we notice that the noncommutative $n\times n$
  Permanent is also uniquely characterized by the identities
  $p_1(x)\equiv x$ and $p_i(X)=\sum_{j=1}^i x_{1j} p_{i-1}(X_j)$ for
  $1<i \leq n$, where $X$ is a matrix of $i^2$ noncommuting variables
  and $X_j$ is its $j$-th minor w.r.t.\ the first row. I.e.\ if
  arbitrary polynomials $p_i, 1\leq i\leq n$ satisfies these $n$
  identities over \emph{noncommuting} variables $x_{ij}, 1\leq i,j\leq
  n$ if and only if $p_i$ computes the $i\times i$ permanent of
  noncommuting variables. The rest of the proof is exactly as in
  Impagliazzo-Kabanets \cite{KI03}. We can easily describe an NP machine
  to simulate a $\P^{\Perm_{\Z}}$ computation. The NP machine guesses
  a polynomial-size noncommutative circuit for $\Perm$ on $m\times m$
  matrices, where $m$ is a polynomial bound on the matrix size of the
  queries made. Then the NP verifies that the circuit computes the
  permanent by checking the $m$ \emph{noncommutative} identities it
  must satisfy. This can be done in deterministic polynomial time by
  assumption. Finally, the NP machines uses the circuit to answer all
  the integer permanent queries. Putting it together, we get
  $\NEXP=\NP$ which contradicts the nondeterministic time hierarchy
  theorem.
\end{proof}

\section{Schwartz-Zippel lemma over finite rings}\label{sz-lemma}

In this section we give a generalization of Schwartz-Zippel Lemma to
finite commutative rings and apply it for identity testing of
black-box polynomials in $R[x_1,\cdots,x_n]$, where $R$ is a finite
commutative ring with unity whose elements are uniformly encoded by
strings from $\{0,1\}^m$ with a special string $e$ denote unity, and
the ring operations are performed by a ring oracle.

We recall some facts about finite commutative rings \cite{B74,AM69}. A
commutative ring $R$ with unity is a \emph{local ring} if $R$ has a
\emph{unique} maximal ideal $M$. An element $r\in R$ is
\emph{nilpotent} if $r^n = 0$ for some positive integer $n$. An
element $r\in R$ is a \emph{unit} if it is invertible.  I.e.\ $rr'=1$
for some element $r'\in R$. Any element of a finite local ring is
either a nilpotent or a unit. An ideal $I$ is a \emph{prime ideal} of
R if $ab\in I$ implies either $a\in I$ or $b\in I$. For finite
commutative rings, prime ideals and maximal ideals coincide. These
facts considerably simplify the study of finite commutative rings (in
contrast to infinite rings).

The \emph{radical} of a finite ring $R$ denoted by $\Rad(R)$ is
defined as the set of all nilpotent elements, i.e 
\[ 
\Rad(R) = \{r\in R ~|~ \exists n>0 ~\mbox{s.t}~ r^n = 0\} 
\]

The radical $\Rad(R)$ is an ideal of $R$, and it is the unique maximum
ideal if $R$ is a local ring. Let $m$ denote the least positive
integer such that for every nilpotent $r\in R$, $r^m=0$, i.e
$(\Rad(R))^m = 0$. Let $R$ be any finite commutative ring with unity
and $\{P_1,P_2,\cdots,P_\ell\}$ by the set of all maximal ideals of
$R$. Let $R_i$ denote the quotient ring $R/P_i^m$ for $1\leq i\leq
\ell$. Then, it is easy to see that each $R_i$ is a local ring and
$P_i/P_i^m$ is the unique maximal ideal in $R_i$. We recall the
following structure theorem for finite commutative rings.

\begin{theorem}[\cite{B74}, Theorem VI.2, page 95]\label{struct-thm}
  Let $R$ be a finite commutative ring. Then $R$ decomposes (up to
  order of summands) uniquely as a direct sum of local rings. More
  precisely
\[ R \iso R_1 \oplus R_2 \oplus \cdots \oplus R_{\ell}, \] via the map
$\phi(r) = (r + P_1^m, r + P_2^m, \cdots, r + P_{\ell}^m)$, where
$R_i=R/P_i^m$ and $P_i, 1\leq i\leq \ell$ are all the maximal ideals of
$R$.  \end{theorem}  

It is easy to see that $\phi$ is a homomorphism with trivial kernel.
The isomorphism $\phi$ naturally extends to the polynomial ring
$R[x_1,x_2,\cdots,x_n]$, and gives the isomorphism $\hat{\phi}:
R[x_1,x_2,\cdots,x_n]\rightarrow \oplus_{i=1}^{\ell}
R_i[x_1,x_2,\cdots,x_n]$.

\subsection{The Schwartz-Zippel lemma}

We observe the following easy fact about zeros of a univariate
polynomial over a ring.

\begin{proposition}\label{univzeros} 
  Let $R$ be an arbitrary commutative ring containing an integral
  domain $D$. If $f\in R[x]$ is a nonzero polynomial of degree $d$
  then $f(a)= 0$ for at most $d$ distinct values of $a\in D$.
\end{proposition}

\begin{proof} 
  Suppose $a_1,a_2,\cdots,a_{d+1}\in D$ are distinct points such that
  $f(a_i)=0, 1\leq i\leq d+1$. Then we can write $f(x)=(x-a_1)q(x)$
  for $q(x)\in R[x]$. Now, dividing $q(x)$ by $x-a_2$ yields
  $q(x)=(x-a_2)q'(x) + q(a_2)$, for some $q'(x)\in R[x]$.  Thus,
  $f(x)=(x-a_1)(x-a_2)q'(x) + (x-a_1)q(a_2)$. Putting $x=a_2$ in this
  equation gives $(a_2-a_1)q(a_2)=0$. But $a_2-a_1$ is a nonzero
  element in $D$ and is hence invertible. Therefore, $q(a_2)=0$.
  Consequently, $f(x)=(x-a_1)(x-a_2)q'(x)$. Applying this argument
  successively for the other $a_i$ finally yields
  $f(x)=g(x)\prod_{i=1}^{d+1} (x - a_i)$ for some nonzero polynomial
  $g(x)\in R[x]$. Since $\prod_{i=1}^{d+1} (x-a_i)$ is a monic
  polynomial, this forces $\deg(f)\geq d+1$ which is a contradiction.
\end{proof}

Consider a polynomial $f\in R[x_1,\cdots,x_n]$. Let $R'$ denote the
ring $R[x_1,\cdots,x_{n-1}]$. Then we can consider $f$ as a univariate
polynomial in $R'[x_n]$ and apply Lemma~\ref{univzeros}, since $R'$
contains the integral domain $D$ that $R$ contains. Now, by an easy
induction argument on the number of variables as in \cite[Lemma
D.3]{TZ}, we can derive the following analog of the Schwartz-Zippel
test for arbitrary commutative rings containing large enough integral
domains.

\begin{lemma}\label{nullstellensatz} 
  Let $R$ be an arbitrary commutative ring containing an integral
  domain $D$.  Let $g\in R[x_1,x_2,\cdots,x_n]$ be any polynomial of
  degree at most $d$.  If $g\not\equiv 0$, then for any finite subset
  $A$ of $D$ we have
\[
\Prob_{a_1\in A,\cdots,a_n\in A}[g(a_1,a_2,\cdots,a_n)=0] \leq\frac{n
d}{|A|}.  
\]
\end{lemma}

In general Lemma~\ref{nullstellensatz} is not applicable, because the
given ring may not contain a large integral domain. We explain how to
get around this problem in the case of finite local commutative rings.
Because of the structure theorem, it suffices to consider local rings.

Let $R$ be a finite local ring with unity given by a ring oracle.
Suppose the characteristic of $R$ is $p^{\alpha}$ for a prime $p$.  If
the elements of $R$ are encoded in $\{0,1\}^m$ then $2^m$ upper bounds
the size of $R$. Let $M>2^m$, to be fixed later in the analysis.  Let
$U=\{c e~|~0\leq c\leq M\}$, where $e$ denotes the unity of $R$. We
will argue that, for a suitable $M$, if we sample $c e$ uniformly from
$U$ then $(c~\mod~p)~e$ is almost uniformly distributed in
$\mathbb{Z}_p e$. Pick $x$ uniformly at random from $\mathbb{Z}_{M}$
and output $x e$. Let $a\in\mathbb{Z}_p$ and $P=\Prob[x\equiv
a~(\mod~p)]$. The $x$ for which $x\equiv a~(\mod~p)$ are $a,a + p,
\cdots, a + p\lfloor\frac{M-a}{p}\rfloor$.  Let
$M'=\lfloor\frac{M-a}{p}\rfloor$.  Then $P=M'+1/M\leq
\frac{1}{p}(1+\frac{2^m}{M})$. Clearly, $P\geq \frac{1}{p}(1 -
\frac{2^m}{M})$. For a given $\epsilon > 0$, choose $M =
2^{m+1}/\epsilon$. Then $\frac{1-\epsilon/2}{p}\leq P\leq
\frac{1+\epsilon/2}{p}$. So $(x~\mod~p) e$ is
$\frac{\epsilon}{2}$-uniformly distributed in $\mathbb{Z}_pe$.

\begin{lemma}\label{sz-ring} 
  Let $R$ be a finite local commutative ring with unity and of
  characteristic $p^{\alpha}$ for a prime $p$. The elements of $R$ are
  encoded using binary strings of length $m$. Let
  $g\in\R[x_1,x_2,\cdots,x_n]$ be a polynomial of degree at most $d$
  and $\epsilon > 0$ be a given constant. If $g\not\equiv 0$, then
\[
  \Prob_{a_1\in U, \cdots, a_n\in U}[g(a_1,a_2,\cdots,a_n)=0] \leq
  \frac{nd}{p}(1+\frac{\epsilon}{2}), 
\] 
where $U=\{c e~|~0\leq c\leq M\}$ and $M > 2^{m+1}/\epsilon$.
\end{lemma}

\begin{proof} 
  Consider the following tower of ideals inside $R$ : 
  \[
R\supseteq
  pR\supseteq p^2R\supseteq \cdots \supseteq p^{\alpha}R=\{0\}.  
  \]
  Let $k$ be the integer such that $g\in p^kR[x_1,\cdots,x_n]\setminus
  p^{k+1}R[x_1,\cdots,x_n]$. Write $g = p^k \hat{g}$. Consider the
  ring, $\hat{I} = \{r\in R~|~p^k r =0\}$. Clearly, $\hat{I}$ is an
  ideal of $R$.  Let $S=R/(\hat{I} + pR)$.  We claim that $\hat{g}$ is
  a nonzero polynomial in $S[x_1,\cdots,x_n]$. Otherwise, let
  $\hat{g}\in (\hat{I}+pR)[x_1,\cdots,x_n]$. Write $\hat{g} = g_1 +
  g_2$, where $g_1\in \hat{I}[x_1,\cdots,x_n]$ and $g_2\in
  pR[x_1,\cdots,x_n]$.  Then $p^k \hat{g} = p^k g_2$ as $p^k g_1 = 0$.
  But $g_2 \in pR[x_1,\cdots,x_n]$, which contradicts the fact that
  $k$ is the largest integer such that $g\in p^k R[x_1,\cdots,x_n]$. 
  Thus $\hat{g}$ is a nonzero polynomial in
  $S[x_1,\cdots,x_n]$.  Now we argue that $S$ contains the finite
  field $\F_p$, and then using the Lemma~\ref{nullstellensatz}, the
  proof of the lemma will follow easily. To see a copy of $\F_p$
  inside $S$, it is enough to observe that $\{i + (\hat{I} +
  pR)~|~0\leq i\leq p-1\}$ as a field is isomorphic to $\F_p$. Clearly
  the failure probability for identity testing of $g$ in
  $R[x_1,\cdots,x_n]$ is upper bounded by the failure probability for
  the identity testing of $\hat{g}$ in $S[x_1,\cdots,x_n]$.  Consider
  the natural homomorphism $\phi : U\rightarrow \F_p$, given by
  $\phi(c e) = c~\mod~p$.  Thus if we sample uniformly from $U$, using
  $\phi$, we can $\frac{\epsilon}{2}$-uniformly sample from $\F_p$.
  Notice that for any $b\in\mathbb{F}_p$,
  $\frac{1-\epsilon/2}{p}\leq\Prob_{x\in\mathbb{Z}_M}[x\equiv
  b~\mod~p] \leq\frac{1+\epsilon/2}{p}$.  Now using the
  Lemma~\ref{nullstellensatz}, we conclude the following :

  \[ \Prob_{a_1\in U,a_2\in U\cdots a_n\in U}[g(a_1,\cdots,a_n)=0]
  \leq \Prob_{b_1\in \F_p\cdots b_n\in\F_p}
  [\hat{g}(b_1,\cdots,b_n)=0]\leq \frac{nd}{p}(1+\frac{\epsilon}{2}),
  \] where $b_i=a_i~(\mod~p)$. The additional factor of
  $(1+\frac{\epsilon}{2})$ comes from the fact that we are only
  sampling $\frac{\epsilon}{2}$-uniformly from $\F_p$. This can be
  easily verified from the proof of Lemma~\ref{nullstellensatz}.
  Hence we have proved the lemma.
\end{proof}
     
\section{Randomized Polynomial Identity Testing over finite
rings}\label{pit-ring}

In this section we study the identity testing problem over finite
commutative ring oracle with unity. For the input polynomial, we
consider both black-box representation and circuit representation.
First we consider the black-box case. Our identity testing algorithm
is a direct consequence of Lemma~\ref{sz-ring}.
 
\begin{theorem}\label{black-box-poly} 
  Let $\R$ (which decomposes into local rings as $\oplus_{i=1}^\ell R_i$) be a
  finite commutative ring with unity given as a oracle. Let the input
  polynomial $f\in R[x_1,\cdots,x_n]$ of degree at most $d$ be given
  via black-box access. Suppose $R_i$'s is of characteristic
  $p_i^{\alpha_i}$. Let $\epsilon > 0$ be a given constant. 
  If $p_i\geq knd$ for all $i$, for some integer
  $k\geq 2$, we have a randomized polynomial time identity test with
  success probability $1-\frac{1}{k}(1+\frac{\epsilon}{2})$.
\end{theorem}

\begin{proof} 
  Consider the natural isomorphism $\hat{\phi} : R[x_1,x_2,\cdots,x_n]
  \rightarrow \oplus_{i=1}^{\ell} R_i[x_1,x_2,\cdots,x_n]$.  Let
  $\hat{\phi}(f) = (f_1,f_2,\cdots,f_{\ell})$. If $f\not\equiv 0$ then
  $f_i\not\equiv 0$ for some $i\in[\ell]$, where $f_i\in
  R_i[x_1,x_2,\cdots,x_n]$. Fix such an $i$. Our algorithm is a direct
  application of Lemma~\ref{sz-ring}. Define $U=\{c e~|~0\leq c\leq
  M\}$, assign values for the $x_i$'s independently and uniformly at
  random from $U$, and evaluate $f$ using the black-box access. The
  algorithm declares $f\not\equiv 0$ if and only if the computed value
  is nonzero. By Lemma~\ref{sz-ring}, our algorithm outputs the
  correct answer with probability 
  $1-\frac{nd}{p_i}(1+\frac{\epsilon}{2})\geq 1-\frac{1}{k}
  (1+\frac{\epsilon}{2})$.
  \footnote{Notice that we have to compute $c e$ using the ring oracle
    for addition in $R$. Starting with $e$, we need to add it $c$
    times.  The running time for this computation can be made
    polynomial in $\log c$ by writing $c$ in binary and applying the
    standard doubling algorithm.}
\end{proof}

The drawback of Theorem~\ref{black-box-poly} is that we get a
randomized polynomial-time algorithm only when $p_i\geq knd$.

However, when the polynomial $f$ is given by an arithmetic circuit we
will get a randomized identity test that works for all finite
commutative rings given by oracle. This is the main result in this
section.  A key idea is to apply the transformation from \cite{AB03}
to convert the given multivariate polynomial to a univariate
polynomial. The following lemma has an identical proof as \cite[Lemma
4.5]{AB03}.

\begin{lemma}\label{univ-subs} 
  Let $R$ be an arbitrary commutative ring and $f\in
  R[x_1,x_2,\cdots,x_n]$ be any polynomial of maximum degree $d$.
  Consider the polynomial $g(x)$ obtained from $f(x_1,x_2,\cdots,x_n)$ by
  replacing $x_i$ by $x^{(d+1)^{i-1}}$ i.e $g(x) = f(x, x^{(d+1)},
  \cdots, x^{(d+1)^{n-1}})$. Then $f\equiv 0$ over $R[x_1,\cdots,x_n]$ if and
  only if $g\equiv 0$ over $R[x]$.
\end{lemma} 

By Lemma~\ref{univ-subs}, it suffices to describe the identity test
for a univariate polynomial in $R[x]$ given by an arithmetic circuit.
Notice that if $\deg(f)=d$ then we can bound $\deg(g)$ by
$d(d+1)^{n-1}$ which we denote by $D$. Our algorithm is simple and
essentially the same as the Agrawal-Biswas identity test over the
finite ring $\mathbb{Z}_n$ \cite{AB03}.

We will randomly pick a monic polynomial $q(x)\in U[x]$ of degree
$\lceil\log O(D)\rceil$.  Then we carry out a division of $f(x)$ by the
polynomial $q(x)$ and compute the remainder $r(x)\in R[x]$. Our
algorithm declares $f$ to be identically zero if and only if $r(x)=0$.
Notice that we will use the structure of the circuit to carry out the
division. At each gate we carry out the division. More precisely, if
the inputs of a $+$ gate are the remainders $r_1(x)$ and $r_2(x)$,
then the output of this $+$ gate is $r_1 + r_2$. Similarly if $r_1$
and $r_2$ are the inputs of a $*$ gate, then we divide $r_1(x)r_2(x)$
by $q(x)$ and obtain the remainder as its output. Crucially, since
$q(x)$ is a monic polynomial, the division algorithm will make sense
and produce unique remainder even if $R[x]$ is not a U.F.D (which is
the case in general).

We now describe the pseudocode of the identity testing algorithm
(Algorithm 2). Our algorithm takes as input an arithmetic circuit $C$
computing a polynomial $f\in R[x_1,x_2,\cdots,x_n]$ of degree at most
$d$ and an $\epsilon > 0$.

\begin{algorithm}
\caption{The Identity Testing algorithm}
\begin{algorithmic}[1]

\Procedure {\tt IdentityTesting}{$C$,$\epsilon$}
 
  \For {$i=1, n$}
     
    \State $x_i\leftarrow x^{(d+1)^{i-1}}$ \Comment{Univariate transformation}
   
      \EndFor    
       
       \State $g(x)\leftarrow C(x,x^{(d+1)},\cdots,x^{(d+1)^{n-1}})$.

       \State $D\leftarrow d(d+1)^{n-1}$. \Comment{The formal degree of $g(x)$ 
        is at most $D$}
       
        \State Choose a monic polynomial $q(x)\in U[x]$ of degree 
$\lceil\log\frac{12D}{1-\epsilon}\rceil$ uniformly at random.

         \State Divide $g(x)$ by $q(x)$ and compute the remainder $r(x)$. 
          \Comment{The division algorithm uses the structure of the circuit.}
       
       \If {$r(x)=0$}
           \State $C$ computes a zero polynomial. 
    
         \Else \State $C$ computes a nonzero polynomial. 
 
       \EndIf 
\EndProcedure
\end{algorithmic}
\end{algorithm}

We will now prove the correctness of the above randomized identity
test in Lemmas~\ref{div}, \ref{crt}, and \ref{irrd-poly}.

\begin{lemma}\label{div} 
  Let $R$ be a local commutative ring with unity and of characteristic
  $p^{\alpha}$ for some prime $p$ and integer $\alpha > 0$. Let $g$ be
  a nonzero polynomial in $R[x]$ such that $g\in p^k R[x]\setminus
  p^{k+1} R[x]$ for $k < \alpha$. Let $\hat{I}=\{r\in R~|~p^k r =
  0\}$, $g = p^k \hat{g}$ where $\hat{g}\not\in pR$ and $q$ is a monic
  polynomial in $R[x]$. If $q$ divides $g$ in $R$, then $q$ divides
  $\hat{g}$ in $R/(\hat{I} + pR)$.
\end{lemma}  

\begin{proof} 
  As $q(x)$ divides $g(x)$ in $R[x]$, we have $g(x) = q(x) q_1(x)$ for
  some polynomial $q_1(x)\in R[x]$. Suppose $\hat{g}(x) = q(x)
  \bar{q}(x) + r(x)$ in $R[x]$ where the degree of $r(x)$ is less than
  the degree of $q(x)$.  Also note that the division makes sense even
  over the ring as $q(x)$ is monic.  We want to show that $r(x)\in
  (\hat{I} + pR)[x]$.  We have the following relation in $R[x]$:

\[ 
g = q q_1 = p^k \hat{g} = p^k q \bar{q} + p^k r.   
\] 

So, $p^k r = q (q_1 - p^k \bar{q})$. If $(q_1 - p^k \bar{q})\not\equiv
0$ in $R[x]$, then the degree of the polynomial $q (q_1 - p^k
\bar{q})$ is strictly more than the degree of $p^k r$ as $q$ is monic
and degree of $q$ is more than the degree of $r$. Thus $(qq_1 - p^k q
\bar{q})\equiv 0$ in $R[x]$ forcing $p^k r = 0$ in $R[x]$. So by the
choice of $\hat{I}$, we have $r(x)\in \hat{I}[x]$. Thus $r(x)\in
(\hat{I} + pR)[x]$.  Notice that in the Lemma~\ref{sz-ring}, we have
already proved that $\hat{g}(x)\not\equiv 0$ in $S[x]$, where
$S=R/(\hat{I} + pR)$.  Also $q$ is nonzero in $S[x]$ as it is a monic
polynomial. Hence we have proved that $q(x)$ divides $\hat{g}(x)$ over
$S[x]$.  
\end{proof}

The following lemma is basically chinese remaindering tailored to our
setting.

\begin{lemma}\label{crt} 
  Let $R$ be a local ring with characteristic $p^{\alpha}$.  Let
  $g(x)\in p^k R[x]\setminus p^{k+1} R[x]$ for some $k\geq 0$.  Let
  $g(x) = p^k \hat{g}(x)$ and $\hat{I} = \{r\in R~|~p^k r=0\}$.
  Suppose $q_1(x),q_2(x)$ are two monic polynomials over $R[x]$ such
  that each of them divides $g$ in $R[x]$.  Moreover, suppose there
  exist polynomials $a(x),b(x)\in R[x]$ such that $a q_1 + b q_2 = 1$
  in $R/(\hat{I} + pR)$. Then $q_1 q_2$ divides $\hat{g}$ in
  $R/(\hat{I} + pR)$.
\end{lemma}

\begin{proof} 
  By the Lemma~\ref{div}, we know that $q_1$ and $q_2$ divide
  $\hat{g}$ in $R/(\hat{I} + pR)$. Let $\hat{g} = q_1 \bar{q}_1$ and
  $\hat{g} = q_2 \bar{q}_2$ in $R/(\hat{I} + pR)$.  Let $\bar{q}_1 =
  q_2 q_3 + r$ in $R/(\hat{I} + pR)$.  So, $\hat{g} = q_1 q_2 q_3 +
  q_1 r$.  Substituting $q_2 \bar{q}_2$ for $\hat{g}$, we get
  $q_2(\bar{q}_2 - q_1 q_3) = q_1 r$.  Multiplying both side by $a$
  and substituting $a q_1(x) = 1 - b q_2$, we get $q_2 [a (\bar{q}_2 -
  q_1 q_3) + b r] = r$.  If $r \not\equiv 0$ in $R/(\hat{I} + pR)$, we
  arrive at a contradiction since $q_2$ is monic and thus the degree
  of $q_2 [a (\bar{q}_2 - q_1 q_3) + b r] $ is more than the degree of
  $r$.
\end{proof}

Let $f(x)$ be a nonzero polynomial in $R[x]$ of degree at most $D$.
The next lemma states that, if we pick a random monic polynomial
$q(x)\in U[x]$ ($U$ is similarly defined as before)of degree $d
\approx \log O(D)$, with high probability, $q(x)$ will not divide
$f(x)$.

\begin{lemma}\label{irrd-poly} 
  Let $R$ be a commutative ring with unity. Suppose $f(x)\in R[x]$ is
  a nonzero polynomial of degree at most $D$ and $\epsilon > 0$ be a
  given constant.  Choose a random monic polynomial $q(x)$ of degree
  $d = \lceil\log \frac{12D}{1-\epsilon}\rceil$ in $U[x]$.  Then with
  probability at least $\frac{1-\epsilon}{4d}$, $q(x)$ will not divide
  $f(x)$ over $R[x]$.\footnote{An alternative proof of this lemma
based on \cite[Lemma 4.7]{AB03} is given in the appendix.} 

\end{lemma}

\begin{proof} 
  Let $R \iso \bigoplus_i R_i$ is the local ring decomposition of $R$.
  As $f$ is nonzero in $R[x]$, there exists $j$ such that
  $f_j=\hat{\phi}_j(f)$ is nonzero in $R_j[x]$.  Clearly, we can lower
  bound the required probability by the probability that
  $q_j=\hat{\phi}_j(q)$ does not divide $f_j$ in $R_j[x]$. Let the
  characteristic of $R_j$ is $p^{\alpha}$.  If $q_j$ divides $f_j$ in
  $R_j[x]$, then it also divides over $R_j/(\hat{I}_j + pR_j)$. It is
  shown in the proof of the Lemma~\ref{sz-ring}, $\F_p \subset
  R_j/(\hat{I}_j + pR_j)$.  Now the number of irreducible polynomials
  in $\F_p$ of degree $d$ is at least $\frac{p^{d} - 2p^{d/2}}{d}$.
  Let $t = \frac{p^{d} - 2p^{d/2}}{d}$.  Let
  $\hat{q}(x)=\sum_{i=0}^{d-1} b_i x^i + x^d\in\F_p[x]$ be a monic
  polynomial. Now if a monic polynomial $P(x)$ of degree $d$ is
  randomly chosen from $U[x]$ then, $\Prob[P(x)\equiv
  \hat{q}(x)~\mod~p]= \frac{\prod_{i=0}^{d-1}\lfloor (M -
    b_i)/p\rfloor+1}{M^d} \geq \frac{1}{p^d}(1-\frac{2^{m}}{M})^{d}$.
  Again, choosing $M > d 2^{m+1}/\epsilon$, we get $\Prob[P(x)\equiv
  \hat{q}(x)~\mod~p]\geq (1-\epsilon/2)/p^d$.

  So, the probability that $q_j$ is an irreducible polynomial in
  $\F_p[x]$ is at least $t (1-\epsilon)/ p^d > (1 - \epsilon)/2d$.
  Let $f_j\in p^k R_j[x]\setminus p^{k+1} R_j[x] $.  So we can write
  $f_j = p^k f'$, where $f'\in R_j[x]\setminus pR_j[x]$. By the
  Lemma~\ref{div}, $q_j$ divides $f'$ in $R/(\hat{I}_j + pR)$.  Also,
  by the Lemma~\ref{crt}, the number of different monic polynomials
  that are irreducible in $\F_p$ and divides $f'$ in $R_j/(\hat{I}_j +
  pR_j)$ is at most $D/d$.  In the sample space for $q$, any monic
  polynomial of degree $d$ in $R_j/(\hat{I}_j + pR_j)$ occurs at most
  $(\frac{M}{p}+1)^d$ times. So the probability that a random monic
  irreducible polynomial $q$ will divide $f$ is at most
  $\frac{(D/d)(\frac{M}{p}+1)^d}{M^d} \leq
  \frac{D}{dp^d}(1+\frac{1}{d})^d < \frac{3D}{d2^d}$.  So a random
  monic polynomial $q\in U[x]$ (which is irreducible in $\F_p$ with
  reasonable probability) will not divide $f(x)$ with probability at
  least $\frac{1-\epsilon}{2d} - \frac{3D}{d p^d} >
  \frac{1-\epsilon}{4d}$ for $d \geq \lceil\log
  \frac{12D}{1-\epsilon}\rceil$.  
\end{proof}

The correctness of Algorithm 2 and its success probability follow
directly from Lemma~\ref{div}, Lemma~\ref{crt} and
Lemma~\ref{irrd-poly}. 

In particular, by Lemma~\ref{irrd-poly}, the
success probability of our algorithm is at least
$\frac{1-\epsilon}{4t}$, where
$t=\lceil\log\frac{12D}{1-\epsilon}\rceil$.  As
$\frac{1-\epsilon}{4t}$ is an inverse polynomial quantity in input
size and the randomized algorithm has one-sided error, we can boost
the success probability by repeating the test polynomially many times.
We summarize the result in the following theorem.

\begin{theorem}\label{univ-pit} 
  Let $R$ be a finite commutative ring with unity given as an oracle
  and $f\in R[x]$ be a polynomial, given as an arithmetic circuit.
  Then in randomized time polynomial in the circuit size and
  $\log~|R|$ we can test whether $f\equiv 0$ in $R[x]$.
\end{theorem}

Randomized polynomial time identity testing for multivariate
polynomials $f\in R[x_1,\cdots,x_n]$ given by arithmetic circuits
follows from Theorem~\ref{univ-pit} and Lemma~\ref{univ-subs}.

\begin{theorem}\label{mult-pit}
  Let $R$ be a commutative ring with unity given as an oracle. Let $f$
  be a polynomial in $R[x_1,x_2,\cdots,x_n]$ of formal degree at most
  $d$, is given by an arithmetic circuit over $R$. Then in randomized
  time polynomial in circuit size and $\log |R|$ we can test whether
  $f\equiv 0$ in $R[x_1,x_2,\cdots,x_n]$.
\end{theorem}

\begin{remark}
  The randomized polynomial-time identity test of Bogdanov and Wee
  \cite{BW05} for noncommutative circuits of polynomially bounded
  degree in $\F\{x_1,\cdots,x_n\}$ for a field $\F$, can be extended
  to such circuits over any commutative ring $R$ with unity, where $R$
  is given by a ring oracle. This follows from the fact that the
  Amitsur-Levitzki theorem is easily seen to hold even in the ring
  $R\{x_1,\cdots,x_n\}$. The easy details are given in the appendix.
\end{remark}

\begin{remark}
  Finally, we note that the results in Section \ref{dit} carry over
  without changes to noncommuting polynomials in
  $R\{x_1,\cdots,x_n\}$, where $R$ is a commutative ring with unity
  given by a ring oracle.
\end{remark}

\newpage

\appendix

\section{Noncommutative identity testing over commutative
  coefficient rings}\label{noncomm-pit}

Here we extend the noncommutative identity testing of Bogdanov and Wee
\cite{BW05} to over $R\{x_1,\cdots,x_n\}$ where $R$ is an arbitrary commutative
ring with unity. Our algorithm is a combination of
Amitsur-Levitzki's theorem and the Theorem~\ref{mult-pit}. We first
briefly discuss the Amitsur-Levitzki's result tailored to our
application and the result of \cite{BW05}. Let $M_k(\F)$ be the
$k\times k$ matrix algebra over $\F$. The following algebraic lemma
was the key result used in \cite{BW05}.

\begin{lemma}\label{AL}\cite{AL50,GZ05}
  $M_k(\F)$ does not satisfy any non-trivial polynomial identity of
  degree $< 2k$.
\end{lemma}

Based on Lemma~\ref{AL}, a noncommutative version of the
Schwartz-Zippel lemma over $\F\{x_1,\cdots,x_n\}$ is described in
\cite{BW05}. We first give an intuitive description of the identity
testing algorithm in \cite{BW05}. Assume $f\in\F\{x_1,\cdots,x_n\}$ is
of degree $d$ and is given by an arithmetic circuit.  Fix $k$ such
that $k>\lceil d/2\rceil$.  Consider a field extension $\F'$ of $\F$
such that $|\F'| >> d$. The idea is to evaluate the circuit on random
$k\times k$ matrices from $M_k(\F')$.  We think each entry of the
matrix as an indeterminate and view the $k^2$ indeterminates as
commuting variables. So at the output of the circuit, we get a
$k\times k$ matrix such that each of its entries are polynomials in
commuting variables. Lemma~\ref{AL} guarantees that $f\equiv 0$ in
$\F\{x_1,\cdots,x_n\}$ if and only if each of the $k^2$ polynomials
computed as the entries of the matrix at the output gate, are
identically zero.  Then we get a lower bound of the success
probability via commutative Schwartz-Zippel lemma.

We give a randomized polynomial time identity testing algorithm over
$R\{x_1,\cdots,x_n\}$ where $R$ is any finite commutative ring with unity and
is given by a ring oracle. Our algorithm is based on the observation
that Lemma~\ref{AL} is valid over $M_k(R)$. For the sake of
completeness, we briefly discuss the proof of the Lemma~\ref{AL}
tailored to $R$. The following fact is the key in proving the
Lemma~\ref{AL}.

\begin{fact}{\rm\cite[page 7]{GZ05}}\label{m-identity}
  Let $A$ be an $\F$-algebra spanned by a set $B$ over $\F$. If the
  algebra $A$ satisfies an identity of degree $k$ in
  $\F\{x_1,\cdots,x_n\}$,
  then it satisfies a multilinear identity of degree $\leq k$.
\end{fact}

We observe that the result of the Fact~\ref{m-identity} holds, even if
$A$ be an algebra over $R$. Proof is analogous to the proof of the
Fact~\ref{m-identity}. Following \cite[page 7]{GZ05}, we call a
polynomial $f$ \emph{multilinear} if every variable occurs with degree
exactly one in every monomial of $f$.
 
\begin{lemma}\label{m-ring}
  Let $A$ be an $R$-algebra such that $A$ satisfies an identity of
  degree $k$.  Then it satisfies a multilinear identity of degree $k$.
\end{lemma} 

\begin{proof} 
 The lemma follows from an identical argument to that in the 
 proof of Theorem 1.3.7 in \cite{GZ05}.
\end{proof}

Using Lemma~\ref{m-ring}, it follows that Lemma~\ref{AL} extends to
$M_k(R)$. The proof is analogous to the proof of Theorem 1.7.2 in
\cite{GZ05}. Let $f$ be an identity for $M_k(R)$ of degree $< 2k$. By
the Lemma~\ref{m-ring}, we can assume that $f$ is multilinear.  Also,
multiplying $f$ by the new variables from the right, we can assume
that the degree of $f$ is $2k - 1$. Let,
\[ 
f(x_1,x_2,\cdots,x_{2k-1}) = \sum_{\sigma\in S_{2k-1}} \alpha_{\sigma}
x_{\sigma(1)} \cdots x_{\sigma(2k-1)}
\]   
with $\alpha_1 \neq 0$, where $1$ denotes the identity permutation.
Let $e_{ij}$ be the $k\times k$ matrix with unity (of $R$) at the
$(i,j)$-th entry and zero in all other entries. It is easy to see that
$f(e_{11},e_{12},e_{22},e_{23}, \cdots, e_{k-1,k}, e_{k k}) = \alpha_1
e_{1 k} \neq 0$, since $x_1\cdots x_{2k - 1}$ is the only monomial
that does not vanish during the evaluation. So $f$ is not an identity
for $M_k(R)$. The fact that $R$ is a ring with unity is crucially
used.

\begin{lemma}\label{AL-1}
  Let $R$ be a finite commutative ring with unity. Then $M_k(\R)$ does
  not satisfy any polynomial identity of degree $< 2k$.
\end{lemma}

Now we a randomized polynomial time identity testing algorithm over 
$R\{x_1,\cdots,x_n\}$.

\begin{theorem}\label{noncomm-ring}
Let $f\in R\{x_1,\cdots,x_n\}$ be a polynomial of degree $d$, given by a
noncommutative arithmetic circuit $C$. $R$ is given as a ring oracle and its 
elements are encoded using binary strings of length $m$. Then there 
is a randomized polynomial time algorithm (\poly(n,d,m)) to test 
if $f\equiv 0$ over $R\{x_1,\cdots,x_n\}$. 

\end{theorem}

\begin{proof}
  Let $x_1,x_2,\cdots,x_n$ are the indeterminates in $C$. Choose
  $k=\lceil d/2\rceil + 1$. Replace each $x_i$ by a $k\times k$ matrix
  over the set of indeterminates $\{y^{(i)}_{j\ell}\}_{1\leq
    j,\ell\leq k}$. Once we replace $x_i$ by matrices , the inputs and
  the outputs of the gates will be matrices. Replace each addition
  (multiplication) gate by a block of circuits computing the sum
  (product) of two $k\times k$ matrices (without loss of generality,
  assume that the fan-in of all gates is two). This can be easily
  achieved using $O(k^2)$ gates. Let $\hat{C}$ be the arithmetic
  circuit obtained from $C$ by these modifications. Clearly, $\hat{C}$
  computes a function from $\F^{nk^2}\rightarrow \F^{k^2}$ and the
  size of $\hat{C}$ is only polynomial in the size of $C$. Denote by
  $\bar{Y}$ the variable list $(y^{(1)}_{11}, \cdots, y^{(1)}_{kk}, \cdots,
  y^{(n)}_{11}, \cdots, y^{(n)}_{kk})$. Then,
\[
\hat{C}(\bar{Y})=(P_1(\bar{Y}), \cdots, P_{k^2}(\bar{Y})). 
\]
Also, by the Lemma~\ref{AL-1}, $M_{k}(R)$ does not satisfy any
identity of degree $< 2k$ over $R\{x_1,\cdots,x_n\}$. So $f$ satisfies
$M_k(R)$ if and only if $f\equiv 0$ in $R\{x_1,\cdots,x_n\}$, which
equivalently implies that $P_i\equiv 0$ over $R[\bar{Y}]$ for all $i$.
Notice that the degree of $P_i$ is $\leq d$. Now we appeal to the
Theorem~\ref{mult-pit} in order to test whether $P_i\equiv 0$ in time
$\poly(n,d,m)$. 
\end{proof}

Bogdanov and Wee in \cite{BW05} evaluate the noncommutative circuit
over a field extension $\F'$ of $\F$ in case $\F$ is a small field
compared to the degree. In our proof of Theorem~\ref{noncomm-ring},
when coefficients come from the ring $R$, we avoid working in a ring
extension and instead apply Theorem~\ref{mult-pit}.

\section{Alternative proof of Lemma~\ref{irrd-poly}}

Let $R$ be a finite commutative ring with unity (denoted $e$)
and its elements uniformly encoded in $\{0,1\}^m$. 

Recall we need to show the following: if we divide a nonzero
polynomial $g(x)\in R[x]$ of degree $D$ by a random monic polynomial
$q(x)\in U[x]$ of degree $\log O(D)$ then with high probability we get
a nonzero remainder. Recall from Section~\ref{pit-ring} that $U=\{ke\mid 0\leq
k\leq M-1\}$, where $M>2^{m+1}/\epsilon$.

Indeed, Agrawal and Biswas essentially show in \cite[Lemma 4.7]{AB03}
that the above result holds for the special case when the ring $R$ is
the ring $\mathbb{Z}_t$ of integers modulo $t$, where $t$ is any
positive integer given in binary. In Section~\ref{pit-ring} we gave a
self-contained proof of Lemma~\ref{irrd-poly}. In the sequel we give a
different proof which applies the \cite{AB03} result for
$\mathbb{Z}_t$ and brings out an interesting property of the division
algorithm.

Let $n$ denote the characteristic of the ring $R$. Then sampling from
$U[x]$ amounts to almost uniform sampling from the copy of $\Z_n[x]$,
namely $\Z_ne[x]$, contained in $R[x]$ as a subring. Since $(R,+)$ is
a finite abelian group, by the fundamental theorem for abelian groups,
we can write $(R,+)$ as a direct sum $R=\bigoplus_{i=1}^k \Z_{n_i}e_i$,
where $e_1,\cdots,e_k$ forms an independent generating set for $(R,+)$,
and $n_i$ is the additive order of $e_i$ for each $i$. Notice that
the lcm of $n_1,\cdots,n_k$ is the ring's characteristic $n$. This
decomposition extends naturally to the additive group $(R[x],+)$ to
give 
\begin{equation}
\label{add_iso_eqn}
R[x]=\bigoplus_{i=1}^k \Z_{n_i}[x]e_i.
\end{equation}
Thus, every polynomial $g(x)\in R[x]$ can be uniquely written as
$g(x)=\sum_{i=1}g_i(x)e_i$, where $g_i$ is a polynomial with
integer coefficients in the range $0,\cdots,n_i-1$ for each $i$. 
Clearly, dividing $g(x)$ by $q(x)$ amounts to dividing each term
in $\sum_{i=1}g_i(x)e_i$. The following claim tells us how to
analyze this term by term division. More precisely, we analyze
the quotient and remainder when we divide $g_i(x)e_i\in R[x]$ 
by $q(x)\in \Z_n[x]$ ($\iso Z_ne[x]\subseteq R[x]$). 

\begin{claim}
Let $g_i(x)=q(x)q'(x)+r'(x)$ be the quotient and remainder
when we divide $g_i(x)$ by $q(x)$ in the ring $\Z_{n_i}[x]$.
Let $g_i(x)e_i=q(x)q''(x)+r''(x)$ be the quotient and
remainder when we divide $g_i(x)e_i$ by $q(x)$ in the ring
$R[x]$. Then $q'(x)e_i=q''(x)$ and $r'(x)e_i=r''(x)$.
\end{claim}
This claim is somewhat surprising because Equation \ref{add_iso_eqn}
only gives us a \emph{group} decomposition of $R[x]$ and not a
\emph{ring} decomposition. Thus, it is not clear why division in the
ring $\Z_{n_i}[x]$ can be related to division in $R[x]$. Indeed, the
crucial reason why we can relate the two divisions is because the
divisor polynomial $q(x)$ lives in the copy of $\Z_n[x]$ inside
$R[x]$.

To see the claim, we will carry out the division of $g_i(x)$ by $q(x)$
over $R[x]$. Since both $g_i$ and $q(x)$ have integer coefficients,
this amounts to carrying out division in $\Z_n[x]$ which yields, say,
$g_i(x)=q(x)q_1(x)+r_1(x)$. We can also write $q_1(x)=a(x)+n_ib(x)$
and $r_1(x)=c(x)+n_id(x)$. Then, over $\Z_{n_i}$, notice that we must
have $g_i(x)=q(x)a(x)+c(x)$. Therefore, $a(x)=q'(x)$ and $c(x)=r'(x)$.
Now, multiplying both sides by $e_i$ we will get
$q_1(x)e_i=a(x)e_i+n_ie_ib(x)=a(x)e_i=q'(x)e_i$. Similarly, we get
$r_1(x)e_i=c(x)e_i=r'(x)e_i$. Furthermore, again multiplying both
sides by $e_i$, we also get $g_i(x)e_i=q(x)q_1(x)e_i+r_1(x)e_i$. Hence,
$q''(x)=q_1(x)e_i=q'(x)e_i$ and $r''(x)=r_1(x)e_i=r'(x)e_i$. This
proves the claim.

A consequence of the claim is the following nice property of the
division algorithm: in order to divide $g(x)$ by $q(x)$ over the ring
$R$, for each $i$ we can carry out the division of $g_i(x)$ by $q(x)$
over the ring $\Z_{n_i}$ and obtain the quotients and remainders:

\[
 g_i(x)=q(x)q'_i(x)+r'_i(x).
\]
Then we can put together the term by term divisions to obtain
\begin{eqnarray}\label{eqn1}
g(x)=q(x)(\sum_{i=1}^k q'_i(x)e_i)+(\sum_{i=1}^k r'_i(x)e_i).
\end{eqnarray}
More precisely, when we divide $g(x)$ by $q(x)$ in $R[x]$, the
quotient is $\sum_{i=1}^k q'_i(x)e_i$ and the remainder is
$\sum_{i=1}^k r'_i(x)e_i$. Now, since $g\in R[x]$ is nonzero, there is
an index $j$ such that $g_j[x]\in \Z_{n_j}[x]$ is nonzero.
Furthermore, since $n_j$ is a factor of $n$, the polynomial $q(x)$
modulo $n_j$ is still an almost uniformly distributed random monic
polynomial. It follows from the Agrawal-Biswas result \cite[Lemma
4.7]{AB03} applied to division of $g_j(x)$ by $q(x)$ over $\Z_{n_j}$
that $r'_j(x)$ will be nonzero with high probability. Consequently, by
Equation~\ref{eqn1} the remainder $\sum_{i=1}^k r'_i(x)e_i$ on
dividing $g(x)$ by $q(x)$ in the ring $R[x]$ is also nonzero with the
same probability.
    
\end{document}